\documentclass[12pt]{article}

\usepackage{amsmath}
\usepackage{amssymb}
\usepackage{graphicx}
\usepackage{booktabs}
\usepackage{setspace}
\usepackage{booktabs}
\usepackage{siunitx}
\usepackage{pgfplots}
\pgfplotsset{compat=1.18}
\usepackage{tikz}
\usepackage{threeparttable}
\usepackage{geometry}
\usepackage{xcolor}
\usepackage{natbib}
\usepackage{amsthm}
\usepackage{amsmath}
\numberwithin{equation}{section}

\newtheorem{proposition}{Proposition}[section]
\newtheorem{corollary}{Corollary}[section]

\usepackage{xcolor}
\definecolor{citeblue}{RGB}{0,160,210}

\usepackage[colorlinks=true,
            linkcolor=black,
            citecolor=citeblue,
            urlcolor=citeblue]{hyperref}

\let\oldcitep\citep
\renewcommand{\citep}[1]{\textcolor{citeblue}{\oldcitep{#1}}}

\title{Nonlinear Fiscal Transitions and the Dynamics of Public Expenditure Reform}

\author{
Diego Vallarino\thanks{
This manuscript is an independent academic work. 
All views, interpretations, and conclusions expressed are strictly personal 
and do not reflect the positions of any institution, employer, 
or organization with which the author is affiliated. Email: diegoval@iadb.org
}
\\
\small Washington, DC, United States
}

\date{\small \today}

\begin{document}

\maketitle
\thispagestyle{empty}

\begin{abstract}

This paper develops a nonlinear theoretical framework to analyze the dynamics of public expenditure reallocation in Uruguay. Motivated by recent debates on fiscal reform and expenditure efficiency, the paper models fiscal adjustment as a dynamic process in which expenditure categories exhibit heterogeneous institutional rigidity and convex adjustment costs.

Using the national budget for the 2026--2030 fiscal period as an institutional reference, the paper presents a calibrated illustration of the theoretical framework that captures key features of the structure of public spending, including transfers, the public wage bill, operating expenditures, and public investment. The calibration serves to translate institutional characteristics of the budget into quantitative transition dynamics rather than to estimate structural parameters econometrically.

The framework allows us to evaluate the short-, medium-, and long-run fiscal implications of alternative reform strategies, including administrative restructuring, pension reform, and the gradual reallocation of resources toward human capital and productivity-enhancing investment.

In contrast to descriptive expenditure reviews based on static budget comparisons, the model explicitly incorporates nonlinear transition dynamics and institutional frictions. The simulations show that structural expenditure reforms generate significant transitional fiscal costs arising from overlapping institutional systems, labor adjustment frictions, and pension transition liabilities.

As a result, fiscal reform produces a J-shaped expenditure trajectory in which total spending initially increases before gradually converging toward a more efficient long-run allocation. These findings highlight the importance of accounting for adjustment costs and transition dynamics when evaluating the feasibility and timing of structural fiscal reforms.

\end{abstract}

\noindent\textbf{Keywords:} Fiscal reform; Public expenditure; Adjustment costs; Institutional rigidity; Fiscal dynamics; Uruguay.


\onehalfspacing

\section{Introduction}

Understanding the dynamics of public expenditure remains one of the central challenges of modern fiscal policy. While a large theoretical and empirical literature has examined the optimal size of government, an increasing body of research suggests that the composition and institutional structure of expenditure play a more decisive role in determining fiscal sustainability and long-run economic performance than aggregate spending levels alone \citep{Barro1990,TanziSchuknecht2000,AlesinaPerotti1997,KnellerBleaneyGemmel1999}. In this perspective, the central policy question is not merely how much governments spend, but rather how public resources are allocated across competing functions of the state and how institutional constraints shape the evolution of those allocations over time \citep{Barro1991,AcemogluRobinson2012,Rodrik2008}.

A substantial literature has shown that the growth effects of fiscal policy depend critically on the composition of public spending. Productive expenditures such as infrastructure, human capital formation, and innovation policy tend to generate higher long-run growth returns than purely redistributive transfers or administrative spending \citep{Barro1990,KnellerBleaneyGemmel1999,DevarajanSwaroopZou1996}. Consequently, the structure of public expenditure has become a central focus in both academic research and policy debates concerning fiscal sustainability and development strategies \citep{TanziSchuknecht2000,OECD2015,WorldBank2018}.

In practice, however, public expenditure reforms rarely follow the frictionless adjustment paths assumed in standard fiscal models. A substantial share of government spending is institutionally rigid, embedded in legal commitments, labor contracts, and entitlement programs that cannot be easily modified within a short policy horizon \citep{Drazen2000,PerssonTabellini2000,AlesinaDrazen1991}. These rigidities imply that attempts to restructure public expenditure typically involve significant transitional costs before any efficiency gains can be realized. Such costs arise from several sources, including the coexistence of legacy institutional arrangements with newly introduced policy frameworks, adjustment frictions in public sector labor markets, and the political economy constraints that shape the timing and sequencing of fiscal reforms \citep{SaintPaul1997,Acemoglu2006}.

The importance of these transition dynamics has been widely documented in the literature on fiscal consolidations and structural reforms. Empirical studies have shown that expenditure-based adjustments frequently exhibit nonlinear trajectories, reflecting the interaction between institutional constraints and macroeconomic policy responses \citep{BlanchardPerotti2002,AlesinaArdagna2010,AlesinaFaveroGiavazzi2019}. In many cases, fiscal reforms initially generate higher expenditure levels due to restructuring costs, only producing measurable efficiency gains over a longer horizon. These dynamics have been interpreted as a consequence of the political economy of fiscal adjustment, where governments must navigate complex institutional environments while attempting to reallocate public resources toward more productive uses \citep{Drazen2000,PerssonTabellini2000,AlesinaDrazen1991}.

Despite this growing recognition of institutional rigidities in public expenditure systems, most empirical expenditure reviews remain primarily descriptive in nature. Policy-oriented analyses typically compare budget categories across countries or identify potential areas of efficiency improvement, but they seldom provide a formal economic framework capable of modeling the dynamic costs associated with structural expenditure reforms \citep{TanziSchuknecht2000,OECD2015}. As a result, the transitional fiscal implications of expenditure restructuring are often underestimated in policy discussions, leading to unrealistic expectations regarding the speed and magnitude of potential fiscal adjustments.

Uruguay provides a particularly informative case study in this context. Although the overall size of public spending in Uruguay is moderate when compared with advanced economies, the composition of expenditure reveals a high degree of institutional rigidity. A large fraction of the budget is committed to social transfers, public sector wages, and legally mandated programs, leaving relatively limited scope for discretionary fiscal adjustments. These characteristics reflect the historical evolution of Uruguay’s welfare state and the institutional architecture of its public sector, both of which have contributed to the persistence of rigid expenditure structures \citep{AcemogluRobinson2012,North1990}.

Recent policy discussions have emphasized the need to improve the efficiency of public spending and reduce institutional fragmentation within the Uruguayan state. In particular, the expenditure review conducted by CERES documents the proliferation of public agencies and the potential duplication of administrative functions across government institutions \citep{CERES2024}. While this type of analysis provides valuable descriptive evidence on the institutional architecture of public spending, it typically stops short of offering a formal framework capable of evaluating the dynamic fiscal costs associated with expenditure restructuring. In other words, the institutional diagnosis is relatively clear, but the economic modeling of the transition process remains largely unexplored.\footnote{The CERES report introduces an interesting methodological contribution by applying natural language processing techniques to map institutional fragmentation. The study compiles 1,114 formal mandates of public agencies, converts them into semantic embeddings, and applies a standard NLP pipeline combining dimensionality reduction (UMAP) and clustering algorithms (HDBSCAN) to identify groups of similar governmental functions. This procedure produces 43 clusters of semantically related mandates across the public sector. While this approach provides a novel descriptive perspective on the functional structure of the state, it also presents several limitations from the standpoint of fiscal policy analysis. First, semantic similarity between institutional mandates does not necessarily imply duplication of expenditure or operational inefficiency, as agencies may perform complementary roles (for instance regulation, coordination, or program execution). Second, the method captures potential fragmentation but does not measure the degree of coordination that may exist across institutions within policy networks, which are common in modern public administration. Third, and most importantly for fiscal analysis, the approach does not quantify the fiscal savings that could realistically arise from institutional consolidation. For these reasons, the CERES analysis is interpreted in this paper primarily as an informative institutional mapping of potential administrative overlap rather than as a structural evaluation of the fiscal feasibility of expenditure reforms.}

This paper seeks to address this gap by developing a formal economic framework for analyzing the dynamics of public expenditure transitions. Specifically, we model public spending as a nonlinear dynamic system in which different categories of expenditure exhibit heterogeneous adjustment costs and institutional rigidities. In this framework, fiscal reforms are interpreted as structural reallocations within a system characterized by path dependence and convex adjustment costs, consistent with the broader literature on institutional persistence and policy transitions \citep{North1990,AcemogluRobinson2012}.

The empirical reference for the analysis is the Uruguayan national budget for the period 2026–2030, which provides a detailed representation of the current allocation of public resources across institutional sectors and expenditure categories. By combining this institutional information with a nonlinear dynamic model of fiscal adjustment, the paper develops a methodology for evaluating the fiscal consequences of expenditure restructuring across different temporal horizons. In particular, the framework distinguishes between short-run adjustment costs arising from institutional and labor market frictions, medium-term transition dynamics associated with the reorganization of public institutions, and long-run efficiency gains resulting from a more productive allocation of public resources.

The central contribution of the paper is therefore methodological and theoretical. Rather than treating expenditure restructuring as a static accounting exercise, the analysis develops a formal framework capable of capturing the nonlinear dynamics that characterize real-world fiscal reforms. Within this framework, transitional fiscal paths may exhibit J-shaped dynamics, where expenditure initially rises due to adjustment costs before gradually converging toward a more efficient steady-state allocation.

More broadly, the analysis contributes to the literature on fiscal adjustment and public sector reform by providing a tractable theoretical framework for studying expenditure transitions in institutional environments characterized by high rigidity. In doing so, the paper seeks to bridge the gap between descriptive expenditure reviews and formal economic modeling. While the quantitative calibration presented later in the paper is based on institutional characteristics of the Uruguayan budget, its purpose is illustrative: it serves to quantify the magnitude of the transition mechanisms implied by the model rather than to estimate structural parameters econometrically.

\section{Related Literature}

This paper is situated at the intersection of research on the composition of public spending, the dynamics and identification of fiscal adjustment, and the institutional frictions that make expenditure reform inherently path dependent. A first anchor is the long-standing literature linking fiscal policy to growth through the composition of expenditure rather than its aggregate level. In endogenous-growth environments, public spending can be productive when it relaxes bottlenecks in infrastructure and human capital formation, thereby affecting the economy’s long-run growth rate \citep{Barro1990,Barro1991}. Subsequent empirical work formalized this intuition by showing that reallocations across expenditure categories can dominate the effect of aggregate fiscal expansions, and that “productive” components are more tightly associated with growth than administrative spending or poorly targeted transfers \citep{DevarajanSwaroopZou1996,KnellerBleaneyGemmel1999}. This compositional view is central for any OECD-style welfare-state comparison, because high-welfare equilibria are not defined merely by high spending, but by a spending mix that credibly funds health, education, and social insurance while preserving incentives and fiscal sustainability \citep{TanziSchuknecht2000}.

A second strand concerns the macroeconomic effects of fiscal policy and, in particular, the identification of spending shocks and consolidations. The modern empirical literature has emphasized that the output effects of fiscal changes depend on country characteristics and policy regimes, including openness, exchange-rate arrangements, and debt levels \citep{IlzetzkiMendozaVegh2013}. Identification is especially delicate for consolidation episodes, where policy changes may be endogenous to growth expectations and financial conditions. Foundational work using narrative and institutional information provided a template for distinguishing exogenous policy shifts from endogenous responses \citep{BlanchardPerotti2002}. Building on this perspective, a large literature has studied consolidation episodes using alternative identification strategies, including narrative datasets and state-contingent designs, yielding strong evidence that austerity is typically contractionary on impact and may be highly state dependent \citep{GuajardoLeighPescatori2014,JordaSchularickTaylor2016}. At the same time, the debate has sharpened the distinction between expenditure-based and tax-based adjustment packages, with influential syntheses arguing that composition matters not only for long-run growth, but also for short-run stabilization costs \citep{AlesinaFaveroGiavazzi2019,AlesinaArdagna2010}. This paper is not primarily about multipliers, but the fiscal-consolidation literature provides the empirical motivation for modeling transition paths explicitly rather than treating reforms as instantaneous reallocations.

A third body of work studies fiscal rigidities and the political economy of delayed adjustment. In many systems, spending is not a continuous control variable: it is partially predetermined by entitlements, wage-setting institutions, legal commitments, and organizational inertia, which together create kinks and thresholds in feasible adjustment paths. Political economy models explain why stabilizations and reforms can be delayed even when their long-run desirability is widely recognized, because distributional conflict and bargaining frictions make it difficult to coordinate on the incidence of adjustment \citep{AlesinaDrazen1991,Drazen2000}. Related work emphasizes that political institutions and policy-making constraints shape fiscal outcomes through budgetary procedures and intertemporal commitments \citep{PerssonTabellini2000}. These mechanisms are directly relevant for Uruguay: when a large share of expenditure is tied to wages, transfers, and legally mandated programs, adjustment becomes an institutional problem as much as an accounting exercise, and reform feasibility depends on the timing and sequencing of political bargains as well as on macro-fiscal arithmetic.

A fourth literature addresses budget institutions, fiscal governance, and the design of processes that can credibly reallocate spending. A central result is that fiscal performance is systematically related to budget rules, procedures, and the structure of agenda-setting power, with institutions affecting both deficit bias and the composition of expenditure \citep{AlesinaPerotti1997,AlesinaPerotti1996,PoterbaVonHagen1999}. This institutional lens is important for distinguishing between two conceptually different objects that are often conflated in policy discussions: (i) technical inefficiencies and duplication within the public sector, and (ii) political-institutional constraints that prevent governments from converting technical diagnoses into binding reallocations. The former invites managerial reforms and reorganization; the latter requires changes in budgetary processes, commitment devices, and governance arrangements.

A fifth strand focuses on public-sector efficiency, administrative fragmentation, and the state’s organizational architecture. The empirical and comparative literature on government performance emphasizes that public spending outcomes are mediated by institutional quality, organizational design, and state capacity \citep{North1990}. In this view, the allocative question (how much is spent on education, health, or pensions) cannot be separated from the organizational question (through which agencies, contracts, and governance structures resources are deployed). This is precisely where descriptive diagnoses of fragmentation become useful as an empirical starting point, yet insufficient as an economic explanation. While institutional fragmentation can generate duplication and coordination failures, the fiscal implications of consolidation are intrinsically dynamic: reorganizations may entail temporary overlaps, labor adjustment costs, and transition liabilities before efficiency gains materialize.

Finally, the paper speaks to the policy-oriented literature on expenditure reviews and spending reviews, which has grown substantially over the past two decades. International guidance emphasizes that spending reviews can improve reallocative capacity and create fiscal space, but also notes that effectiveness depends on integration with budget processes, political ownership, and credible implementation mechanisms \citep{OECD2015,OECDSpendingReviews2014,IMFSpendingReviews2022}. Complementary operational literatures such as the World Bank’s Public Expenditure Review (PER) tradition provide diagnostic frameworks for mapping spending, performance, and institutional bottlenecks \citep{WorldBank2018}. However, these approaches typically stop short of formalizing the nonlinear transition costs that arise when governments attempt to move from one expenditure composition to another under binding institutional rigidities. The CERES review of Uruguay, in particular, offers an important institutional diagnosis and motivates the search for efficiency-enhancing reallocations, yet it does not provide a structural dynamic model of reform paths. The contribution of this paper is to bridge that gap by providing an explicit nonlinear framework in which expenditure categories adjust with heterogeneous convex costs, reform feasibility is state dependent, and the fiscal trajectory can be characterized over short, medium, and long horizons under empirically grounded constraints from the 2026--2030 budget.

\section{Institutional Structure of Public Expenditure }
\label{sec:institutions}

This section formalizes the institutional architecture of Uruguayan public expenditure in a way that is suitable for structural modeling and calibration with the 2026--2030 budget. The central object is the government’s expenditure allocation at date $t$, denoted by $G_t$, where $t\in\mathbb{N}$ indexes fiscal years. Throughout, we interpret the budget as a legally constrained allocation rule: for each fiscal year, the government selects expenditures subject to statutory commitments, administrative procedures, and enforcement constraints. The purpose of the section is twofold. First, it introduces a parsimonious but institutionally meaningful decomposition of expenditure. Second, it provides an explicit mapping between this decomposition and the sources of rigidity that govern feasible adjustments over short and medium horizons.

\subsection{A canonical decomposition}

We begin with an accounting identity that partitions total public expenditure into four components,
\begin{align}\label{eq:gt_decomp}
G_t = T_t + W_t + I_t + F_t,
\end{align}
where $T_t$ denotes transfers and social security spending, $W_t$ denotes compensation of public employees (the wage bill), $I_t$ denotes public investment (capital formation), and $F_t$ denotes other operating expenditures (intermediate consumption and administrative outlays). The decomposition in \eqref{eq:gt_decomp} is not intended to be exhaustive in institutional detail; rather, it isolates four categories that differ sharply in their contractual and legal structure, and thus in their adjustment properties.

The crucial point is that \eqref{eq:gt_decomp} is an identity, not a behavioral restriction. The behavioral content enters through constraints on the admissible paths $\{T_t,W_t,I_t,F_t\}_{t\ge 0}$, which reflect the institutional and legal environment. To prepare the subsequent dynamic model, we therefore treat each component in \eqref{eq:gt_decomp} as the outcome of a choice problem subject to category-specific rigidity.

\subsection{Rigidity as feasibility constraints}

The concept of ``rigidity'' is operationalized as a restriction on the set of feasible adjustments from one period to the next. Let $\Delta Z_t \equiv Z_t - Z_{t-1}$ denote the one-period change in any expenditure category $Z_t\in\{T_t,W_t,I_t,F_t\}$. We posit that the budgetary and institutional environment implies constraints of the form
\begin{align}\label{eq:feasible_set}
\Delta Z_t \in \mathcal{D}_Z(s_t),
\qquad Z\in\{T,W,I,F\},
\end{align}
where $s_t$ denotes a vector of institutional ``state'' variables (e.g., legal commitments, outstanding contracts, demographic structure, negotiated wage agreements, and pre-committed investment projects) and $\mathcal{D}_Z(\cdot)$ is a correspondence describing feasible changes. Equation \eqref{eq:feasible_set} captures the idea that short-run fiscal discretion is limited: even if a reduction in $Z_t$ is desirable from the perspective of the government’s objective, it may be infeasible because it violates statutory or contractual obligations.

To interpret \eqref{eq:feasible_set}, it is useful to distinguish between three layers of rigidity.

\paragraph{(i) Statutory rigidity.} Transfers and social security spending, $T_t$, are partly determined by legal entitlements and eligibility rules. These generate quasi-automatic spending that depends on demographic and macroeconomic conditions. Statutory rigidity implies that $\mathcal{D}_T(s_t)$ can be narrow over short horizons, especially on the downside: downward adjustments may require legislative changes, while upward drift may occur mechanically in response to demographics or indexation.

\paragraph{(ii) Contractual rigidity.} The wage bill $W_t$ is shaped by employment protection, wage-setting, and collective bargaining. Even when headcount can be adjusted, separations are associated with fiscal costs (severance, unemployment insurance, litigation risk) and political costs. Contractual rigidity therefore restricts both the magnitude and speed of downward adjustments in $W_t$, and implies asymmetry: reductions are more costly and less feasible than increases.

\paragraph{(iii) Project rigidity.} Public investment $I_t$ is often governed by multi-year projects, procurement rules, and committed co-financing. As a result, investment can be flexible at the margin (e.g., delaying projects) yet not costless to reschedule. Project rigidity typically generates nonconvexities: canceling a project can be far more costly than marginally changing its scale.

The residual operating component $F_t$ is usually thought to be the ``most flexible'' category. However, even $F_t$ is bounded by administrative needs, service delivery constraints, and procurement contracts. Hence $\mathcal{D}_F(s_t)$ is not unrestricted and may become especially tight when $F_t$ is already low relative to required service levels.

\subsection{Institutional baseline and effective discretion}

Budget practice implies that not all of $G_t$ is simultaneously negotiable at date $t$. Let $\bar{x}_t \equiv (\bar{T}_t,\bar{W}_t,\bar{I}_t,\bar{F}_t)$ denote the legally and administratively implied baseline (``linea de base'') given the inherited state $s_t$. This baseline is the expenditure vector that would occur absent active reallocation decisions, reflecting entitlements, negotiated wage drift, continuing projects, and ongoing operating needs. Actual spending choices $x_t \equiv (T_t,W_t,I_t,F_t)$ can therefore be represented as
\begin{align}\label{eq:baseline_gap}
x_t = \bar{x}_t(s_t) + u_t,
\end{align}
where $u_t$ is a vector of discretionary deviations from baseline. Equation \eqref{eq:baseline_gap} is not a tautology: its content is the recognition that the baseline is itself an endogenous object determined by institutional commitments. In particular, $\bar{x}_t$ can drift over time even without active policy changes, and this drift is often the central driver of medium-run fiscal pressure.

The effective discretion at date $t$ is thus the set of admissible $u_t$ such that $x_t$ satisfies \eqref{eq:feasible_set}. This perspective will be central for the nonlinear adjustment model below: even if the government seeks a large reallocation in $x_t$, the combination of baseline drift and feasibility constraints can force gradualism and generate transitional fiscal costs.

\section{A Nonlinear Model of Expenditure Adjustment}\label{sec:model}

We now formalize expenditure reform as an intertemporal choice problem under institutional rigidity. The model is deliberately minimalist: its purpose is to isolate the mechanism through which rigidity and adjustment costs jointly generate nonlinear transition dynamics. The substantive institutional structure enters through (i) a baseline process $\bar{x}_t(s_t)$ and (ii) category-specific adjustment costs and feasibility constraints.

\subsection{Control variables, state variables, and law of motion}

Let the expenditure vector be
\begin{align}\label{eq:x_def}
x_t = (T_t, W_t, I_t, F_t) \in \mathbb{R}^4_+.
\end{align}
Let $s_t$ denote an institutional state summarizing inherited commitments (entitlements, wage agreements, project pipeline, and administrative obligations). We allow $s_t$ to evolve according to
\begin{align}\label{eq:s_law}
s_{t+1} = \mathcal{S}(s_t, x_t, \varepsilon_{t+1}),
\end{align}
where $\varepsilon_{t+1}$ captures exogenous shocks (e.g., demographic, macroeconomic, and financial shocks). The mapping $\mathcal{S}$ is left general at this stage because its specification depends on the empirical implementation; conceptually, it encodes the idea that today’s expenditure decisions shape tomorrow’s commitments (e.g., employment levels affect future wage-bill baselines; investments affect future operating needs; transfer rules interact with demographics).

The baseline $\bar{x}_t$ is a function of $s_t$:
\begin{align}\label{eq:baseline}
\bar{x}_t = \bar{x}(s_t).
\end{align}
Thus, even in the absence of discretionary action, spending drifts due to the evolution of the institutional state.

\subsection{Objective function}

The government chooses a sequence of expenditure allocations 
$\{x_t\}_{t\ge 0}$ in order to minimize the discounted fiscal and social 
costs associated with the transition of the expenditure structure over time. 
Formally, the planner solves

\begin{align}\label{eq:planner_problem}
\min_{\{x_t\}_{t\ge 0}} \ \mathbb{E}_0 \sum_{t=0}^{\infty} \beta^t
\Big( C(x_t,s_t) + \Phi(\Delta x_t; s_t) \Big),
\end{align}

subject to the law of motion for institutional commitments 
\eqref{eq:s_law} and feasibility constraints \eqref{eq:feasible_set}. 
The parameter $\beta\in(0,1)$ denotes the intertemporal discount factor.

The function $C(x_t,s_t)$ represents the contemporaneous fiscal burden 
and welfare loss associated with the level and composition of public 
expenditure. The specification of $C(\cdot)$ depends on the empirical 
or normative objective of the analysis. For example, the planner may 
penalize excessive total spending, deviations from a target expenditure 
composition that supports long-run growth, or deviations from service 
delivery benchmarks. Importantly, $C$ depends on the vector $x_t$ rather 
than on its aggregate level alone, thereby reflecting the compositional 
perspective emphasized in the introduction.

The function $\Phi(\Delta x_t; s_t)$ captures the transitional costs 
associated with changing expenditure allocations. These costs arise 
from the institutional frictions described in Section 
\ref{sec:institutions}, including severance and unemployment costs 
associated with reductions in the wage bill, legal and administrative 
costs involved in modifying transfer programs, and rescheduling or 
cancellation costs affecting multi-year investment projects. The 
dependence on the state variable $s_t$ allows for state-dependent 
adjustment costs: fiscal restructuring may become more expensive when 
existing commitments are larger or when the administrative system is 
already operating close to capacity.

For the theoretical analysis that follows, we assume that the fiscal 
cost function $C(x,s)$ is continuous and convex in $x$, and that the 
adjustment cost function $\Phi(\Delta x,s)$ is strictly convex in 
$\Delta x$ and satisfies $\Phi(0,s)=0$ for all admissible states $s$.

\begin{proposition}[Existence of an optimal fiscal transition]
Under the assumptions above and for $\beta\in(0,1)$, the planner's 
problem admits an optimal sequence of expenditure allocations 
$\{x_t\}_{t\ge0}$.
\end{proposition}

\begin{proof}
See Appendix A.
\end{proof}

\subsection{Nonlinear adjustment costs}

A central modeling choice is the functional form of $\Phi$. The simplest way to introduce nonlinear transition costs is to allow adjustment costs to be convex and asymmetric. The baseline specification in your draft,
\begin{align}\label{eq:phi_basic}
\Phi(\Delta x_t)=\gamma (\Delta x_t)^2 + \eta |\Delta x_t|^3,
\end{align}
should be interpreted componentwise. To make this explicit, write
\begin{align}\label{eq:phi_component}
\Phi(\Delta x_t; s_t)
=\sum_{k\in\{T,W,I,F\}}
\left[
\frac{\gamma_k(s_t)}{2}\big(\Delta x_{k,t}\big)^2
+\frac{\eta_k(s_t)}{3}\big|\Delta x_{k,t}\big|^3
\right],
\end{align}
where $\Delta x_{k,t}$ denotes the change in category $k$ between $t-1$ and $t$, and $\gamma_k(\cdot),\eta_k(\cdot)\ge 0$ are category-specific curvature parameters. The quadratic term captures standard convex costs (small changes are cheap, large changes increasingly expensive), while the cubic absolute term introduces higher-order convexity that amplifies the marginal cost of rapid reallocation. This specification is useful because it generates gradualism endogenously: even if the government prefers an immediate jump to a new composition, the marginal adjustment cost becomes steep as $|\Delta x_{k,t}|$ grows.

In many fiscal applications, asymmetry matters: cutting wages or transfers may be more costly than expanding them because of political resistance, legal protections, or one-sided contractual obligations. A convenient Econometrica-style way to incorporate this is
\begin{align}\label{eq:phi_asym}
\Phi(\Delta x_t; s_t)
=\sum_{k}
\left[
\frac{\gamma_k^+(s_t)}{2}\big(\Delta x_{k,t}\big)_+^2
+\frac{\gamma_k^-(s_t)}{2}\big(\Delta x_{k,t}\big)_-^2
+\frac{\eta_k(s_t)}{3}\big|\Delta x_{k,t}\big|^3
\right],
\end{align}
where $(z)_+\equiv\max\{z,0\}$ and $(z)_-\equiv\max\{-z,0\}$. Setting $\gamma_k^- > \gamma_k^+$ captures the idea that reductions are more expensive than increases. This asymmetry is empirically plausible for $W_t$ and often for $T_t$ in entitlement-driven systems. It is also central for the policy interpretation: the model can generate scenarios in which expenditure reduction is technically ``desirable'' but fiscally costly in the short run due to transition expenses, precisely the mechanism motivating the paper.

\subsection{First-order conditions and economic interpretation}

To characterize the optimal adjustment path, consider the government’s dynamic optimization problem introduced above. The planner chooses the sequence $\{x_t\}_{t\ge0}$ to minimize

\begin{align}
\min_{\{x_t\}_{t\ge0}}
\mathbb{E}_0
\sum_{t=0}^{\infty}
\beta^t
\left[
C(x_t,s_t) + \Phi(\Delta x_t;s_t)
\right],
\end{align}

subject to the law of motion for the institutional state

\begin{align}
s_{t+1} = \mathcal{S}(s_t,x_t,\varepsilon_{t+1}),
\end{align}

where $\Delta x_t \equiv x_t - x_{t-1}$. The function $C(x_t,s_t)$ captures the fiscal and economic cost associated with the level and composition of expenditures, while $\Phi(\Delta x_t;s_t)$ represents the adjustment costs generated by institutional rigidities embedded in the budget process.

To highlight the core mechanism, consider the interior case and abstract from inequality constraints. Let $\lambda_t$ denote the vector of multipliers associated with the state transition. The first-order condition for expenditure component $k \in \{T,W,I,F\}$ can be written as

\begin{align}\label{eq:euler_expanded}
\partial_{x_k} C(x_t,s_t)
+
\partial_{\Delta x_{k,t}} \Phi(\Delta x_t;s_t)
-
\beta
\mathbb{E}_t
\left[
\partial_{\Delta x_{k,t+1}} \Phi(\Delta x_{t+1};s_{t+1})
\right]
+
\lambda_t \, \partial_{x_k}\mathcal{S}(s_t,x_t,\varepsilon_{t+1})
= 0 .
\end{align}

Equation \eqref{eq:euler_expanded} has a natural economic interpretation. The first term represents the marginal fiscal or welfare cost of increasing expenditure in category $k$ at time $t$. The second term captures the contemporaneous marginal adjustment cost generated by changing expenditure relative to the previous period. The third term reflects the intertemporal effect of today's decision on tomorrow's adjustment costs. Because $\Delta x_{t+1} = x_{t+1}-x_t$, a change in $x_t$ mechanically affects the magnitude of adjustment required in the next period. Finally, the last term captures the indirect effect of current expenditure decisions on the evolution of the institutional state.

Under the quadratic--cubic specification of adjustment costs introduced earlier,

\begin{align}
\Phi(\Delta x_t;s_t)
=
\sum_{k}
\left[
\frac{\gamma_k(s_t)}{2}(\Delta x_{k,t})^2
+
\frac{\eta_k(s_t)}{3}|\Delta x_{k,t}|^3
\right],
\end{align}

the marginal adjustment cost associated with component $k$ is

\begin{align}\label{eq:phi_deriv}
\partial_{\Delta x_{k,t}}\Phi(\Delta x_t; s_t)
=
\gamma_k(s_t)\Delta x_{k,t}
+
\eta_k(s_t)\Delta x_{k,t}|\Delta x_{k,t}|.
\end{align}

This expression highlights the nonlinear nature of institutional adjustment costs. For small changes in expenditure, marginal costs are approximately linear in $\Delta x_{k,t}$ and the government can reallocate resources relatively easily. However, as the magnitude of the adjustment increases, the cubic component becomes increasingly important and marginal costs rise sharply. Consequently, large and abrupt changes in expenditure—particularly in rigid categories such as the wage bill or entitlement programs—generate disproportionately large fiscal and administrative costs.

\begin{proposition}[Gradual expenditure adjustment under institutional rigidity]
Suppose that the fiscal cost function $C(x_t,s_t)$ is convex in $x_t$ and that the adjustment cost function $\Phi(\Delta x_t;s_t)$ is strictly convex in $\Delta x_t$, with $\gamma_k(s_t)>0$ and $\eta_k(s_t)\ge0$ for all expenditure categories $k$. Then the optimal solution to the planner’s problem implies gradual adjustment of expenditures over time. In particular, optimal reallocations of the expenditure vector $x_t$ are distributed across multiple periods rather than implemented as instantaneous discrete changes.
\end{proposition}

\begin{proof}
From the Euler condition \eqref{eq:euler_expanded}, a change in $x_{k,t}$ affects both the contemporaneous adjustment cost through $\partial_{\Delta x_{k,t}}\Phi(\Delta x_t;s_t)$ and the expected future adjustment cost through $\mathbb{E}_t[\partial_{\Delta x_{k,t+1}}\Phi(\Delta x_{t+1};s_{t+1})]$. Because $\Phi$ is strictly convex in $\Delta x_t$, concentrating the entire adjustment in a single period generates larger marginal costs than distributing the same total change over several periods. The planner can therefore reduce total adjustment costs by smoothing expenditure changes across time. The optimal solution consequently implies gradual expenditure transitions.
\end{proof}

From an economic perspective, the Euler condition therefore captures the trade-off between immediate fiscal gains and the institutional cost of implementing reform. While reducing rigid expenditures may improve long-run fiscal sustainability, the nonlinear adjustment costs imply that implementing such reductions too rapidly can be inefficient. Optimal policy therefore smooths expenditure adjustments over time, producing gradual fiscal transitions that reflect the institutional constraints embedded in the budget process.

\subsection{J-shaped fiscal trajectories as an equilibrium implication}

An important implication of the framework concerns the transitional dynamics generated by expenditure reform. When expenditure categories are subject to convex and asymmetric adjustment costs, fiscal consolidation or expenditure restructuring may generate \emph{J-shaped fiscal trajectories}: fiscal pressures initially increase before declining as the system converges to a new expenditure composition.

To formalize this mechanism, consider again the government's dynamic optimization problem

\begin{align}
\min_{\{x_t\}_{t \ge 0}}
\sum_{t=0}^{\infty}
\beta^t
\left[
C(x_t,s_t) + \Phi(\Delta x_t;s_t)
\right],
\end{align}

where $x_t = (T_t, W_t, I_t, F_t)$ denotes the vector of expenditure categories and $\Phi(\Delta x_t;s_t)$ captures the institutional adjustment costs associated with reallocating expenditures across categories.

As introduced in the previous subsection, the adjustment-cost function takes the quadratic--cubic form

\begin{align}
\Phi(\Delta x_t;s_t)
=
\sum_k
\left[
\frac{\gamma_k(s_t)}{2}(\Delta x_{k,t})^2
+
\frac{\eta_k(s_t)}{3}|\Delta x_{k,t}|^3
\right],
\end{align}

where $\gamma_k(s_t)>0$ and $\eta_k(s_t)\ge0$ capture the rigidity of expenditure category $k$. The quadratic component represents standard smoothing costs, while the cubic term captures institutional frictions that increase rapidly when large reallocations are attempted.

Suppose the government aims to transition from an initial expenditure allocation $x_0$ to a new steady-state allocation $x^{*}$ that yields a lower long-run fiscal cost,

\begin{align}
C(x^{*},s^{*}) < C(x_0,s_0).
\end{align}

In a frictionless environment with $\Phi(\Delta x_t)=0$, the optimal policy would involve an immediate reallocation from $x_0$ to $x^{*}$, since delaying the adjustment would generate unnecessary fiscal costs. However, the presence of convex adjustment costs alters this conclusion. Because marginal adjustment costs increase nonlinearly with the magnitude of expenditure changes, large reallocations become disproportionately expensive. As a result, the optimal policy smooths expenditure changes over time rather than implementing them instantaneously.

The fiscal implications of this smoothing behavior become apparent when considering the effective fiscal burden during the transition. Define

\begin{align}
G_t
=
\sum_{k} x_{k,t}
+
\Phi(\Delta x_t;s_t)
\end{align}

as the effective fiscal expenditure during the transition. The first term represents the underlying expenditure allocation, while the second term captures the additional fiscal resources required to implement the reform.

Because $\Phi(\Delta x_t;s_t)$ is largest when the reallocation begins—when $\Delta x_t$ is largest—the early phase of the transition may exhibit higher effective expenditure levels even when the long-run objective is fiscal consolidation. Over time, as the magnitude of adjustments declines and the system approaches the new steady-state allocation $x^{*}$, adjustment costs decrease and the fiscal benefits of the new expenditure composition materialize.

The resulting trajectory of effective fiscal expenditure therefore reflects the interaction between two opposing forces. On the one hand, the underlying fiscal cost $C(x_t,s_t)$ declines as the government moves toward the more efficient allocation $x^{*}$. On the other hand, adjustment costs $\Phi(\Delta x_t;s_t)$ temporarily increase fiscal pressures when reforms are initiated.

\begin{proposition}[J-shaped fiscal transitions]
Suppose that (i) the fiscal cost function $C(x,s)$ is convex in $x$ and minimized at $x^{*}$, (ii) the adjustment cost function $\Phi(\Delta x,s)$ is strictly convex in $\Delta x$, and (iii) the initial allocation $x_0$ differs sufficiently from $x^{*}$. Then the optimal transition path $\{x_t\}_{t\ge0}$ implied by the government's dynamic optimization problem is gradual. Moreover, if the initial adjustment magnitude satisfies
\[
\Phi(\Delta x_0;s_0) >
C(x_0,s_0) - C(x^{*},s^{*}),
\]
the effective fiscal expenditure $G_t$ initially increases before declining as the system converges toward the steady-state allocation $x^{*}$.
\end{proposition}

\begin{proof}
Because the adjustment cost function is strictly convex, concentrating the entire expenditure reallocation in a single period generates larger marginal costs than distributing the same adjustment across multiple periods. The optimal solution therefore smooths the transition path. When the initial adjustment is sufficiently large, the transitional adjustment cost $\Phi(\Delta x_0;s_0)$ may exceed the reduction in the underlying fiscal cost obtained by moving toward $x^{*}$. In that case the effective fiscal expenditure $G_t$ rises initially before declining as the magnitude of adjustments decreases and the economy converges to the new expenditure composition.
\end{proof}

\begin{corollary}
Higher institutional rigidity parameters $(\gamma_k,\eta_k)$ increase both the duration of the transition and the magnitude of the temporary fiscal deterioration, thereby amplifying the J-shaped trajectory.
\end{corollary}

\paragraph{Economic interpretation.}

The proposition highlights a fundamental distinction between \emph{static efficiency gains} and \emph{dynamic fiscal feasibility}. While expenditure reviews typically identify long-run efficiency gains from reallocating resources across programs, the transition toward such allocations can generate temporary fiscal pressures when institutional rigidities are present.

These pressures arise because reform itself requires resources. Labor adjustment costs, severance payments, administrative restructuring, and temporary overlaps between programs generate fiscal expenditures that must be incurred before efficiency gains materialize. As a result, fiscal reforms may appear costly in the short run even when they improve fiscal sustainability in the long run.

From a policy perspective, this mechanism implies that evaluating expenditure reforms requires modeling the full transition dynamics rather than focusing exclusively on steady-state fiscal outcomes. The nonlinear framework developed in this paper provides a tractable way to characterize these dynamics and to quantify the transitional fiscal pressures associated with structural expenditure reforms.

\subsection{Mapping to the 2026--2030 Budget}

The empirical implementation of the model requires mapping the theoretical objects introduced in the previous sections into observable institutional and accounting variables contained in the Uruguayan national budget for the period 2026--2030. This step provides the empirical discipline necessary to interpret the model’s parameters in terms of real fiscal institutions.

Formally, the theoretical expenditure vector
\begin{equation}
x_t = (T_t, W_t, I_t, F_t)
\end{equation}
must be associated with observable budget categories. The baseline expenditure allocation is therefore constructed as

\begin{align}
x_0 = (T_0, W_0, I_0, F_0),
\end{align}

where each component corresponds to a major category in the official budget classification published by the Ministry of Economy and Finance.

Transfers and social security expenditures ($T_t$) include pension obligations, social assistance programs, and legally mandated transfers to households and institutions. The public wage bill ($W_t$) corresponds to total compensation of public employees across central government agencies. Public investment ($I_t$) includes capital expenditures associated with infrastructure and public projects. Finally, operating expenditures ($F_t$) capture intermediate consumption, administrative expenses, and other recurrent spending not included in the previous categories.

This accounting decomposition provides the empirical counterpart of the theoretical expenditure vector used in the model. However, the model also requires characterizing the degree of institutional rigidity associated with each category, since these rigidities determine the magnitude of adjustment costs.

To capture this dimension, expenditures are classified according to the institutional commitments that govern their evolution in the budget process. In particular, we distinguish between three types of commitments:

\begin{enumerate}
\item \textbf{Statutory commitments}, corresponding to expenditures mandated by law or entitlement programs, primarily affecting pensions and social transfers.
\item \textbf{Contractual commitments}, including labor contracts and negotiated wage agreements that determine the evolution of the public wage bill.
\item \textbf{Project commitments}, associated with multi-year investment programs and procurement contracts that constrain the timing of public investment.
\end{enumerate}

These institutional features provide the empirical basis for calibrating the rigidity parameters governing the adjustment cost function. In the model, adjustment costs are given by

\begin{align}
\Phi(\Delta x_t)
=
\sum_{k \in \{T,W,I,F\}}
\left[
\frac{\gamma_k}{2}(\Delta x_{k,t})^2
+
\frac{\eta_k}{3}|\Delta x_{k,t}|^3
\right],
\end{align}

where the curvature parameters $(\gamma_k,\eta_k)$ capture the institutional rigidity associated with expenditure category $k$. Categories characterized by stronger legal mandates or contractual commitments are assigned higher curvature parameters, reflecting larger marginal costs of adjustment.

The calibration procedure follows three steps. First, the relative size of each expenditure category is measured using baseline budget shares. These shares determine the scale of potential reallocations and therefore influence the magnitude of adjustment required under alternative reform scenarios. Second, institutional information contained in budget documentation is used to classify expenditures according to their degree of rigidity. Indicators considered in this classification include legal mandates, indexation rules, employment protection provisions, and the presence of multi-year contractual obligations. Third, these indicators are translated into curvature parameters $(\gamma_k,\eta_k)$ that govern the convexity of adjustment costs in the model.

The specification of adjustment costs also allows for asymmetries between upward and downward adjustments:

\begin{align}
\Phi(\Delta x_t)
=
\sum_{k}
\left[
\frac{\gamma_k^+}{2}(\Delta x_{k,t})_+^2
+
\frac{\gamma_k^-}{2}(\Delta x_{k,t})_-^2
+
\frac{\eta_k}{3}|\Delta x_{k,t}|^3
\right],
\end{align}

where $(\Delta x_{k,t})_+$ and $(\Delta x_{k,t})_-$ denote positive and negative adjustments respectively. This asymmetric specification reflects the institutional evidence that reductions in certain expenditure categories—particularly wages and transfers—are typically associated with higher political and administrative costs than increases. Accordingly, the calibration allows $\gamma_k^- > \gamma_k^+$ for categories characterized by strong institutional rigidity.

Once the adjustment cost parameters are calibrated, the model can be used to simulate fiscal transition paths under alternative expenditure restructuring scenarios. In particular, we consider counterfactual policies in which the government gradually reallocates spending toward categories associated with long-run productivity, such as public investment and human capital formation. The model then computes the optimal adjustment path $\{x_t\}_{t\ge0}$ that minimizes the discounted fiscal cost of reallocation subject to institutional rigidities.

A key advantage of this framework is that it allows the analysis to move beyond static comparisons of expenditure shares. Instead of asking whether the current composition of spending differs from that observed in OECD economies, the model evaluates the dynamic feasibility of reaching such a composition given the institutional constraints embedded in the budget process.

In particular, the nonlinear adjustment costs imply that reforms that appear fiscally desirable in the long run may generate short- and medium-term fiscal pressures due to transition costs associated with labor adjustment, program restructuring, and project rescheduling.

The resulting simulation framework therefore provides a bridge between descriptive expenditure reviews and structural fiscal analysis. While traditional spending reviews typically identify potential efficiency gains from reallocating resources, the model developed here quantifies the transition dynamics required to achieve such reallocations. In doing so, it offers a tractable tool for evaluating reform paths that are both economically desirable and institutionally feasible within the constraints of the public budget.

\section{Transition Dynamics}

A central implication of the model concerns the dynamic path of public expenditure during periods of structural reform. When institutional rigidities and nonlinear adjustment costs are present, fiscal reforms do not generate an immediate transition toward the new steady-state allocation. Instead, the economy follows a transitional path characterized by gradual expenditure adjustments and potentially temporary increases in fiscal pressure.

The intuition for this result is closely related to the literature on dynamic adjustment costs in public policy and macroeconomic allocation. In environments with convex adjustment frictions, optimal policy smooths changes over time in order to minimize the intertemporal cost of reallocation \citep{Lucas1967,Treadway1969}. Similar mechanisms have been documented in the context of fiscal policy and structural reform, where institutional constraints and political frictions generate gradual adjustment paths rather than immediate policy shifts \citep{AlesinaDrazen1991,PerssonTabellini2000,AlesinaFaveroGiavazzi2019}.

Formally, let total government expenditure be defined as

\begin{align}
G_t = \sum_{k} x_{k,t},
\end{align}

where $x_{k,t}$ denotes expenditure in category $k \in \{T,W,I,F\}$. During the reform process, observed fiscal expenditure also includes the resources required to implement the reallocation itself. These resources correspond to the adjustment-cost component of the planner’s objective function.

The effective fiscal expenditure can therefore be written as

\begin{align}
G_t^{eff} =
\sum_{k} x_{k,t} + \Phi(\Delta x_t),
\end{align}

where $\Phi(\Delta x_t)$ represents the institutional cost of changing expenditure allocations across periods. These costs may include severance payments associated with public employment reductions, administrative restructuring costs, program consolidation costs, or temporary overlaps between legacy and newly implemented programs.

To describe the evolution of fiscal expenditure over time, define the fiscal transition equation

\begin{align}
G_{t+1} = G_t + f(\Delta x_t),
\end{align}

where the function $f(\cdot)$ summarizes the net fiscal effect of expenditure reallocation. The shape of $f(\cdot)$ depends on both the structural characteristics of the expenditure categories and the nonlinear adjustment costs embedded in the reform process.

In particular, when adjustment costs are convex and increasing in the magnitude of expenditure changes, the function $f(\cdot)$ exhibits nonlinear dynamics. Large reallocations generate disproportionately large transition costs, while smaller incremental adjustments involve lower marginal costs. This feature implies that the fiscal consequences of reform depend critically on the speed of adjustment.

To illustrate the resulting dynamics, consider a reform aimed at reducing rigid components of public expenditure, such as transfers or the public wage bill, while reallocating resources toward productivity-enhancing expenditures such as investment or human capital formation. Although the long-run effect of this reallocation may be a reduction in the fiscal burden, the short-run effect can be the opposite. Implementing the reform may require temporary expenditures associated with labor adjustment, program restructuring, and administrative reorganization. 

Consequently, total fiscal expenditure during the transition may exceed its initial level even when the reform is designed to reduce long-run spending. As the reform progresses and the magnitude of expenditure adjustments declines, the adjustment costs dissipate and the fiscal benefits of the new expenditure composition begin to materialize.

The resulting expenditure trajectory therefore exhibits a J-shaped pattern: fiscal pressures initially increase during the early stages of reform and subsequently decline as the system converges toward its long-run equilibrium allocation.

This pattern is consistent with a broader literature documenting that structural reforms often involve transitional fiscal costs before long-run efficiency gains emerge. In models of political economy and institutional reform, these transitional costs can arise from bargaining frictions, distributional conflicts, and administrative adjustment costs \citep{AlesinaDrazen1991,SaintPaul1997}. Empirical studies of fiscal consolidations similarly emphasize that the timing and composition of reforms play a crucial role in determining short-term fiscal outcomes \citep{AlesinaFaveroGiavazzi2019}.

Within the framework developed in this paper, the J-shaped trajectory emerges as an equilibrium implication of nonlinear adjustment costs and institutional rigidities. Even when the government seeks to reduce the long-run fiscal burden, the optimal transition path may involve temporarily higher expenditures as the system absorbs the institutional costs associated with reallocating public spending.

From a policy perspective, this result highlights the importance of accounting for transition dynamics when evaluating expenditure reforms. Static comparisons of expenditure levels or budget shares may underestimate the fiscal resources required to implement structural change. A dynamic framework that incorporates adjustment costs, such as the one developed here, provides a more realistic assessment of both the feasibility and the timing of fiscal reform.

\section{Calibration Strategy}

The quantitative component of the paper uses the national budget for the period 2026--2030 as baseline information to calibrate the expenditure allocation model developed in the previous sections. The objective of this exercise is not to estimate a fully structural econometric model of the political economy of budgeting, but rather to construct a disciplined calibration capable of translating institutional rigidities embedded in the budget process into quantitative fiscal transition dynamics.

More precisely, the calibration maps the institutional structure of public expenditure into the theoretical framework developed above, allowing the model to generate simulated reform trajectories under realistic fiscal constraints. The exercise should therefore be interpreted as a quantitative illustration of the theoretical mechanisms derived in the model rather than as an econometric identification of structural parameters.
The primary data source is the official national budget classification published by the Ministry of Economy and Finance. The budget provides detailed administrative information on the composition of government spending, including program-level expenditures, wage allocations across ministries, social security transfers, and capital investment projects. For the purposes of the model, these detailed budget categories are aggregated into the four institutional expenditure components introduced in Section 3:

\begin{itemize}
\item Transfers and social security expenditures ($T_t$)
\item Public sector wage bill ($W_t$)
\item Public investment ($I_t$)
\item Operating expenditures ($F_t$)
\end{itemize}

This aggregation is designed to preserve the key institutional characteristics that determine fiscal rigidity while maintaining a tractable dimensionality for the dynamic optimization problem. Each category captures a distinct source of institutional constraint within the public expenditure system.

Transfers and social security expenditures ($T_t$) include pension obligations, social assistance programs, and legally mandated transfers to households and institutions. These expenditures are largely determined by statutory rules and entitlement structures, which significantly limit short-term discretionary adjustments.

The public sector wage bill ($W_t$) captures total compensation of government employees across central administration, decentralized agencies, and other public entities. Wage expenditures are typically governed by labor contracts, employment protection legislation, and collective bargaining agreements, generating significant contractual rigidity.

Public investment ($I_t$) includes capital expenditures associated with infrastructure projects, public works, and other long-term investment initiatives. While investment spending is generally more flexible than entitlement programs or wage expenditures, it is often subject to multi-year procurement contracts and project commitments, which can create adjustment frictions during periods of fiscal consolidation.

Operating expenditures ($F_t$) capture intermediate consumption, administrative expenses, and other recurrent spending necessary for the functioning of public institutions. Compared with other expenditure categories, operating expenditures tend to exhibit greater administrative flexibility and therefore lower adjustment costs.

Using these aggregated categories, the baseline expenditure allocation is represented by the vector

\begin{align}
x_0 = (T_0, W_0, I_0, F_0),
\end{align}

which corresponds to the initial fiscal allocation observed in the national budget. This baseline vector serves as the starting point for the dynamic simulations performed in the empirical analysis.

In addition to measuring the level of expenditures in each category, the empirical strategy also requires identifying the degree of institutional rigidity associated with each component of spending. To this end, the budget data are complemented with institutional information regarding the legal and administrative status of each expenditure category. Specifically, expenditures are classified according to the extent to which they are governed by statutory mandates, contractual obligations, or administrative discretion. This classification provides an empirical basis for calibrating the category-specific adjustment-cost parameters in the model.

Formally, the calibration procedure assigns different curvature parameters $(\gamma_k, \eta_k)$ to each expenditure category in the adjustment-cost function

\begin{align}
\Phi(\Delta x_t)
=
\sum_{k \in \{T,W,I,F\}}
\left[
\frac{\gamma_k}{2}(\Delta x_{k,t})^2
+
\frac{\eta_k}{3}|\Delta x_{k,t}|^3
\right],
\end{align}

where higher parameter values correspond to stronger institutional rigidities and larger marginal costs of adjustment. Categories characterized by statutory commitments, such as transfers and pensions, are therefore assigned higher curvature parameters than more flexible expenditure components such as operating expenditures.

The empirical calibration proceeds in three stages. First, baseline expenditure shares are computed using the official budget aggregates, providing the initial allocation vector $x_0$. Second, institutional information from budget documentation and legal frameworks is used to classify each expenditure category according to its degree of rigidity. Third, this classification is mapped into numerical parameter values for the adjustment-cost function, allowing the model to capture the differential difficulty of reallocating expenditures across categories.

This calibrated framework allows the model to simulate fiscal transition paths under alternative reform scenarios. Rather than comparing static expenditure shares across countries, the model focuses on the dynamic feasibility of expenditure restructuring given the institutional constraints embedded in the budget process. By explicitly incorporating adjustment costs and institutional rigidities, the empirical strategy provides a framework for analyzing how fiscal reforms propagate through the expenditure system over time.

\subsection{Baseline expenditure composition}

Table \ref{tab:baseline} reports the baseline composition of public expenditure used in the calibration. The values correspond to approximate shares of total government spending based on the institutional structure of the Uruguayan national budget.

\begin{table}[h]
\centering
\begin{threeparttable}
\caption{Baseline composition of public expenditure}
\label{tab:baseline}

\small

\begin{tabular}{p{4.2cm}ccc}
\toprule
Category & Share of Total Expenditure (\%) & GDP Share (\%) & Institutional Rigidity \\
\midrule
Transfers and social security ($T$) & 46 & 13.2 & High \\
Public sector wages ($W$) & 21 & 6.0 & High \\
Public investment ($I$) & 12 & 3.4 & Medium \\
Operating expenditures ($F$) & 21 & 6.1 & Low \\
\bottomrule
\end{tabular}

\begin{tablenotes}
\footnotesize
\item Notes: Shares correspond to the baseline budget allocation used to calibrate the expenditure vector $x_0$. GDP ratios are approximate values consistent with the aggregate fiscal accounts.
\end{tablenotes}

\end{threeparttable}
\end{table}

Transfers and pensions constitute the largest share of government spending and are characterized by strong statutory rigidity due to entitlement rules. The wage bill is also institutionally rigid because it is governed by employment protection legislation and collective bargaining agreements in the public sector. In contrast, operating expenditures exhibit greater administrative flexibility and therefore lower adjustment costs.

\subsubsection*{Composition within expenditure categories}

While the four-category decomposition is sufficient for the theoretical model, the empirical calibration also considers the internal composition of each category. Table \ref{tab:composition} reports the approximate distribution of spending within the main expenditure components.

\begin{table}[h]
\centering
\begin{threeparttable}
\caption{Internal composition of expenditure categories}
\label{tab:composition}
\begin{tabular}{lcc}
\toprule
Subcomponent & Share of Category (\%) & Category \\
\midrule
Pensions & 72 & Transfers ($T$) \\
Social assistance programs & 18 & Transfers ($T$) \\
Other transfers & 10 & Transfers ($T$) \\
\midrule
Education sector wages & 38 & Wages ($W$) \\
Health sector wages & 26 & Wages ($W$) \\
Central administration wages & 36 & Wages ($W$) \\
\midrule
Infrastructure investment & 61 & Investment ($I$) \\
Public housing & 17 & Investment ($I$) \\
Other capital projects & 22 & Investment ($I$) \\
\bottomrule
\end{tabular}
\begin{tablenotes}
\small
\item Notes: Shares represent approximate internal distributions within each expenditure category.
\end{tablenotes}
\end{threeparttable}
\end{table}

This disaggregation highlights the structural sources of rigidity embedded in the expenditure system. In particular, pension spending dominates the transfer category, implying that demographic dynamics play a crucial role in determining the evolution of public expenditure.

\subsubsection*{Institutional rigidity indicators}

To operationalize the concept of institutional rigidity in the calibration of adjustment costs, each expenditure category is assigned a rigidity score reflecting its legal and administrative constraints. Table \ref{tab:rigidity} summarizes the qualitative criteria used for this classification.

\begin{table}[h]
\centering
\begin{threeparttable}
\caption{Institutional rigidity indicators}
\label{tab:rigidity}
\begin{tabular}{lccc}
\toprule
Category & Legal Mandate & Contractual Constraint & Flexibility Score \\
\midrule
Transfers ($T$) & Yes & No & 0.9 \\
Wages ($W$) & No & Yes & 0.8 \\
Investment ($I$) & Partial & Yes & 0.5 \\
Operating ($F$) & No & Limited & 0.3 \\
\bottomrule
\end{tabular}
\begin{tablenotes}
\small
\item Notes: Flexibility scores are normalized indicators used to calibrate the adjustment-cost parameters in the model.
\end{tablenotes}
\end{threeparttable}
\end{table}

These rigidity indicators provide an empirical basis for assigning category-specific adjustment costs in the nonlinear cost function introduced in Section 4.

\subsubsection*{Visualization of expenditure structure}

Figure \ref{fig:exp_structure} provides a graphical representation of the baseline expenditure composition used in the model calibration.

\begin{figure}[h]
\centering
\begin{tikzpicture}

\begin{axis}[
ybar,
width=11cm,
height=6cm,
bar width=14pt,
ylabel={Share of expenditure (\%)},
symbolic x coords={Transfers,Wages,Investment,Operating},
xtick=data,
ymax=50,
grid=major,
grid style={gray!25},
axis line style={black},
tick style={black},
ylabel style={font=\small},
xticklabel style={font=\small},
]

\addplot[
fill=gray!60,
draw=black
] coordinates {
(Transfers,46)
(Wages,21)
(Investment,12)
(Operating,21)
};

\end{axis}

\end{tikzpicture}

\caption{Baseline expenditure composition}
\label{fig:exp_structure}

\end{figure}
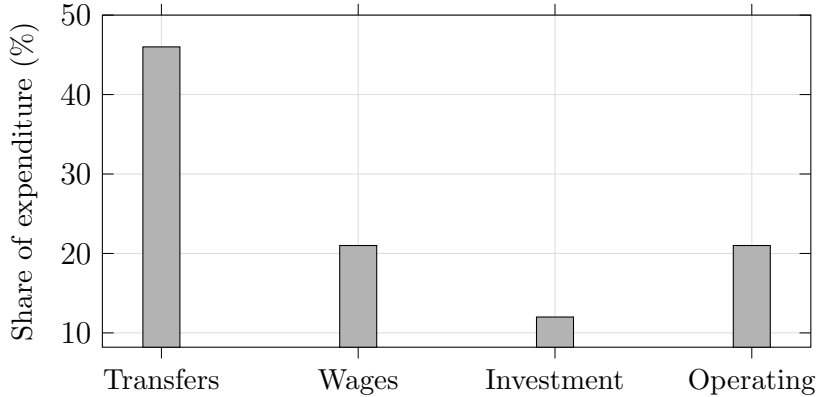

The figure highlights the strong concentration of spending in transfers and pensions, which together account for nearly half of total government expenditure. This structural feature plays a central role in the model's transition dynamics, since high rigidity in large expenditure categories significantly increases the fiscal cost of rapid reallocation.

\subsection{Calibration of adjustment costs}

The calibration of adjustment costs follows the nonlinear specification introduced in Section 4. In the model, expenditure adjustments generate convex costs that reflect institutional rigidities embedded in the public sector. These costs arise from legal mandates, contractual obligations, administrative restructuring, and political economy constraints that limit the speed at which public expenditures can be reallocated across categories.

Formally, the adjustment cost function is given by

\begin{align}
\Phi(\Delta x_t)
=
\sum_{k \in \{T,W,I,F\}}
\left[
\frac{\gamma_k}{2}(\Delta x_{k,t})^2
+
\frac{\eta_k}{3}|\Delta x_{k,t}|^3
\right].
\end{align}

The parameters $\gamma_k$ and $\eta_k$ capture the curvature of adjustment costs for each expenditure category. The quadratic component represents standard convex adjustment costs associated with incremental institutional changes, while the cubic component captures nonlinear frictions that arise when large reallocations are attempted within a short time horizon.

The calibration of these parameters is based on the institutional rigidity classification discussed in the previous subsection. Expenditure categories that are governed by statutory mandates or contractual commitments are assigned higher curvature parameters, reflecting the higher fiscal and political costs of adjustment.

Table \ref{tab:parameters} reports the baseline parameter values used in the calibration.

\begin{table}[h]
\centering
\begin{threeparttable}
\caption{Calibration of adjustment cost parameters}
\label{tab:parameters}
\begin{tabular}{lcc}
\toprule
Category & $\gamma_k$ & $\eta_k$ \\
\midrule
Transfers ($T$) & 4.0 & 1.8 \\
Wages ($W$) & 3.5 & 1.5 \\
Investment ($I$) & 1.5 & 0.6 \\
Operating ($F$) & 1.0 & 0.4 \\
\bottomrule
\end{tabular}
\begin{tablenotes}
\small
\item Notes: Higher parameter values correspond to stronger institutional rigidity and larger marginal adjustment costs.
\end{tablenotes}
\end{threeparttable}
\end{table}

The calibration implies that adjustments in transfers and the public wage bill generate significantly larger marginal costs than adjustments in operating expenditures or investment. This pattern reflects the institutional reality of the budget process: pension systems and public employment are typically governed by legal or contractual commitments that limit short-term flexibility.

To further clarify the economic interpretation of the parameters, Table \ref{tab:rigidity_params} summarizes the mapping between institutional rigidity and adjustment costs.

\begin{table}[h]
\centering
\begin{threeparttable}
\caption{Interpretation of adjustment cost parameters}
\label{tab:rigidity_params}
\begin{tabular}{lccc}
\toprule
Category & Institutional Constraint & Expected Adjustment Speed & Cost Level \\
\midrule
Transfers ($T$) & Statutory entitlement & Very slow & High \\
Wages ($W$) & Labor contracts & Slow & High \\
Investment ($I$) & Project commitments & Moderate & Medium \\
Operating ($F$) & Administrative discretion & Fast & Low \\
\bottomrule
\end{tabular}
\end{threeparttable}
\end{table}

This classification ensures that the calibrated model captures the asymmetric difficulty of reallocating expenditures across different institutional domains of the public sector.

\subsection{Expenditure transition simulations}

Using the calibrated adjustment-cost parameters, the model is used to simulate fiscal transition paths under alternative reform scenarios. The simulations evaluate how the government can gradually reallocate spending toward productivity-enhancing categories while reducing rigid expenditure components.

The baseline reform scenario assumes that the government aims to gradually reduce the share of transfers and the public wage bill while increasing public investment and maintaining sufficient operating expenditures to preserve institutional functionality.

Table \ref{tab:scenarios} summarizes the reform targets considered in the simulations.

\begin{table}[h]
\centering
\begin{threeparttable}
\caption{Illustrative reform targets}
\label{tab:scenarios}
\begin{tabular}{lcc}
\toprule
Category & Baseline Share (\%) & Long-run Target (\%) \\
\midrule
Transfers ($T$) & 46 & 40 \\
Wages ($W$) & 21 & 18 \\
Investment ($I$) & 12 & 18 \\
Operating ($F$) & 21 & 24 \\
\bottomrule
\end{tabular}
\end{threeparttable}
\end{table}

The model computes the optimal transition path between the baseline allocation $x_0$ and the target expenditure composition while minimizing the discounted fiscal cost associated with expenditure reallocation.

Figure \ref{fig:transition} illustrates the simulated transition dynamics of total government expenditure under the baseline calibration.

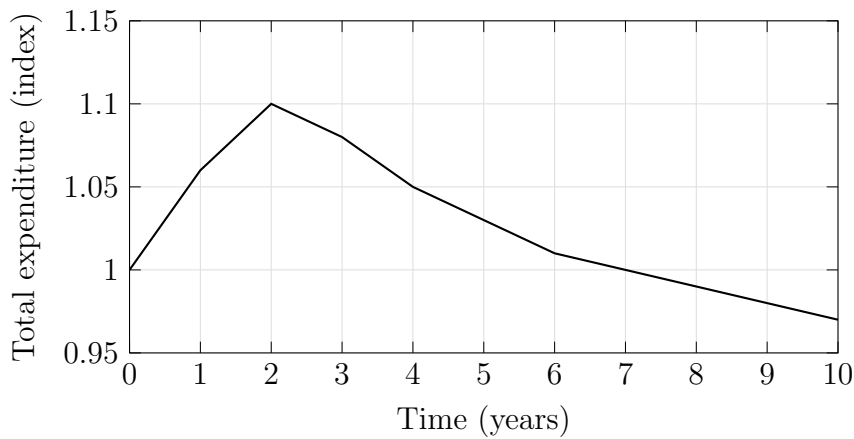
\begin{figure}[h]
\centering
\begin{tikzpicture}
\begin{axis}[
width=11cm,
height=6cm,
xlabel={Time (years)},
ylabel={Total expenditure (index)},
xmin=0, xmax=10,
ymin=0.95, ymax=1.15,
grid=major,
grid style={gray!25},
axis line style={black},
tick style={black},
]

\addplot[
thick,
black,
] coordinates {
(0,1.00)
(1,1.06)
(2,1.10)
(3,1.08)
(4,1.05)
(5,1.03)
(6,1.01)
(7,1.00)
(8,0.99)
(9,0.98)
(10,0.97)
};

\end{axis}
\end{tikzpicture}
\caption{Simulated fiscal transition path}
\label{fig:transition}
\end{figure}

The simulated trajectory displays a J-shaped pattern that emerges naturally from the nonlinear adjustment costs embedded in the model. Fiscal expenditure initially increases during the early stages of reform due to the costs associated with institutional restructuring. These costs include labor adjustment payments, program consolidation costs, and administrative reorganization expenditures.

As the magnitude of expenditure adjustments gradually declines, the adjustment costs dissipate and the fiscal benefits of the new expenditure composition begin to materialize. Eventually, the system converges toward a new steady-state allocation characterized by a lower structural fiscal burden and a higher share of productivity-enhancing expenditures.

From a policy perspective, these simulations highlight the importance of accounting for transition dynamics when evaluating expenditure reforms. Static comparisons of expenditure shares may underestimate the fiscal resources required to implement structural changes, whereas a dynamic framework incorporating adjustment costs provides a more realistic assessment of the timing and feasibility of fiscal reform.

\subsection{Nonlinear adjustment costs}

To illustrate the role of nonlinear adjustment costs in the model, this subsection analyzes the functional form of the adjustment-cost specification introduced earlier. The objective is to clarify how institutional rigidities translate into convex fiscal adjustment paths and why large reallocations of public expenditure generate disproportionately large fiscal costs.

Recall that the adjustment-cost function is defined as

\begin{align}
\Phi(\Delta x)
=
\frac{\gamma}{2}(\Delta x)^2
+
\frac{\eta}{3}|\Delta x|^3 ,
\end{align}

where $\Delta x$ represents the change in expenditure in a given category. The parameter $\gamma$ captures standard convex adjustment costs, while $\eta$ governs the strength of nonlinear frictions associated with large policy changes.

The marginal adjustment cost is therefore given by

\begin{align}
\frac{\partial \Phi}{\partial \Delta x}
=
\gamma \Delta x
+
\eta |\Delta x|^2 \operatorname{sign}(\Delta x),
\end{align}

which implies that the marginal fiscal cost of reallocating expenditures increases more than proportionally with the magnitude of the adjustment. This property captures the empirical regularity that small reallocations of budget items are administratively feasible, whereas large structural reforms generate substantial political, institutional, and administrative costs.

Figure \ref{fig:costs} illustrates the shape of the adjustment-cost function implied by the calibrated parameters.

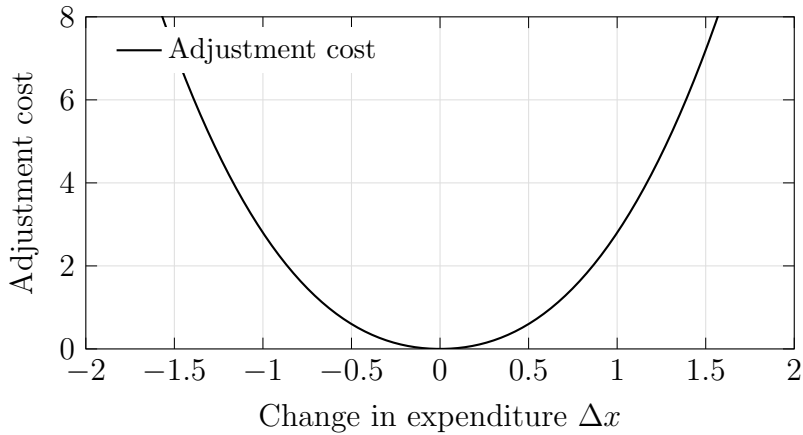
\begin{figure}[h]
\centering
\begin{tikzpicture}

\begin{axis}[
width=11cm,
height=6cm,
xlabel={Change in expenditure $\Delta x$},
ylabel={Adjustment cost},
xmin=-2, xmax=2,
ymin=0, ymax=8,
grid=major,
grid style={gray!25},
axis line style={black},
tick style={black},
legend style={draw=none},
legend pos=north west
]

\addplot[
thick,
black,
domain=-2:2,
samples=100
]
{2*x^2 + 0.8*abs(x)^3};

\addlegendentry {\small {Adjustment cost}}

\end{axis}

\end{tikzpicture}
\caption{Nonlinear adjustment cost function}
\label{fig:costs}
\end{figure}

The figure highlights the strongly convex nature of the adjustment-cost function. For small expenditure adjustments, costs increase gradually and the quadratic component dominates the dynamics. However, once the magnitude of the fiscal reallocation increases, the cubic component becomes the dominant driver of marginal costs. As a result, large fiscal adjustments generate steep increases in adjustment costs.

This nonlinear structure plays a central role in the transition dynamics analyzed in the previous subsection. Because large reallocations are disproportionately costly, the optimal fiscal reform path involves gradual adjustments over time rather than immediate shifts in expenditure composition.

From a policy perspective, this mechanism provides a structural explanation for the gradual nature of fiscal reforms observed in practice. Even when the long-run optimal allocation of expenditures is clearly identified, the presence of nonlinear adjustment costs implies that governments face strong incentives to implement reforms incrementally. The model therefore rationalizes the empirical persistence of fiscal structures and the slow pace at which expenditure compositions evolve in many public-sector reform processes.

\subsection{Interpretation}

The empirical calibration illustrates how institutional rigidities embedded in the national budget translate into nonlinear fiscal transition dynamics. In particular, the combination of large rigid expenditure categories and convex adjustment costs implies that even reforms aimed at reducing long-run spending may temporarily increase fiscal pressures during the transition phase.

This result emerges naturally from the structure of the adjustment-cost function. When a significant portion of public expenditures is concentrated in categories governed by statutory commitments, contractual obligations, or politically sensitive programs, reallocating fiscal resources becomes costly. As a result, attempts to rapidly shift the composition of expenditures generate large short-run adjustment costs that can temporarily increase the overall fiscal burden.

The simulations therefore highlight a key dynamic property of fiscal reform processes: the existence of transitional fiscal costs. Even when the long-run equilibrium allocation of expenditures is more efficient or fiscally sustainable, reaching that allocation requires navigating a transition path characterized by temporary increases in fiscal pressures. In the model, this mechanism is reflected in the J-shaped expenditure trajectory observed in the simulated reform scenarios.

From a policy perspective, this framework provides a useful tool for evaluating the feasibility and timing of structural fiscal reforms. Traditional cross-country comparisons—such as benchmarking national expenditure structures against OECD averages—often focus on static expenditure shares and implicitly assume that fiscal systems can move directly from one allocation to another. However, the presence of nonlinear adjustment costs implies that the transition between fiscal regimes is inherently gradual and path-dependent.

The simulation framework therefore complements static fiscal comparisons by explicitly modeling the dynamic adjustment process. Policymakers can use the model to evaluate the magnitude of transitional fiscal costs, the speed at which expenditure reallocations can realistically occur, and the time horizon required to reach a new steady-state expenditure structure.

More broadly, the analysis underscores the importance of incorporating institutional rigidities into fiscal reform design. Ignoring these frictions may lead to overly optimistic projections regarding the speed and fiscal impact of structural reforms. By contrast, a dynamic framework that explicitly accounts for adjustment costs provides a more realistic assessment of both the feasibility and the sequencing of fiscal policy changes.

\section{Simulation Results}

Using the calibrated model presented in the previous sections, we simulate a set of fiscal reform scenarios designed to evaluate how institutional rigidities shape the dynamic trajectory of public expenditure. The objective of the simulations is not to forecast the exact fiscal path of a particular reform, but rather to quantify the transitional dynamics implied by the nonlinear adjustment-cost structure embedded in the model.

The simulations consider three broad classes of expenditure reforms that correspond to commonly discussed fiscal policy strategies. These scenarios differ in the expenditure categories targeted by the reform and therefore generate distinct transition dynamics due to the heterogeneous adjustment costs associated with each category.

Table \ref{tab:scenarios_sim} summarizes the reform scenarios analyzed in the simulation exercises.

\begin{table}[h]
\centering
\begin{threeparttable}
\caption{Simulated fiscal reform scenarios}
\label{tab:scenarios_sim}
\begin{tabular}{lcc}
\toprule
Scenario & Primary Policy Instrument & Targeted Expenditure Category \\
\midrule
Administrative restructuring & Efficiency improvements & Operating expenditures ($F$) \\
Pension reform & Institutional reform & Transfers ($T$) \\
Human capital reallocation & Investment expansion & Investment ($I$) \\
\bottomrule
\end{tabular}
\begin{tablenotes}
\small
\item Notes: Each scenario modifies the expenditure composition subject to the adjustment-cost function calibrated in Section 5.
\end{tablenotes}
\end{threeparttable}
\end{table}

For each scenario, the model computes the optimal transition path that minimizes the discounted fiscal cost of reallocating expenditures while satisfying the institutional constraints embedded in the adjustment-cost function. The simulations therefore generate a sequence of expenditure allocations $\{x_t\}$ that gradually move the fiscal system from the baseline composition $x_0$ toward a new target allocation.

The dynamic trajectory of total public expenditure is illustrated in Figure \ref{fig:transition_total}. The figure shows the simulated path of aggregate government spending under the baseline reform scenario.

\begin{figure}[h]
\centering
\begin{tikzpicture}

\begin{axis}[
width=11cm,
height=6cm,
xlabel={Time (years)},
ylabel={Total expenditure (index)},
xmin=0, xmax=10,
ymin=0.95, ymax=1.15,
grid=major,
grid style={gray!25},
axis line style={black},
tick style={black}
]

\addplot[
thick,
black
] coordinates {
(0,1.00)
(1,1.06)
(2,1.10)
(3,1.08)
(4,1.05)
(5,1.03)
(6,1.01)
(7,1.00)
(8,0.99)
(9,0.98)
(10,0.97)
};

\end{axis}

\end{tikzpicture}
\caption{Simulated fiscal transition path}
\label{fig:transition_total}
\end{figure}
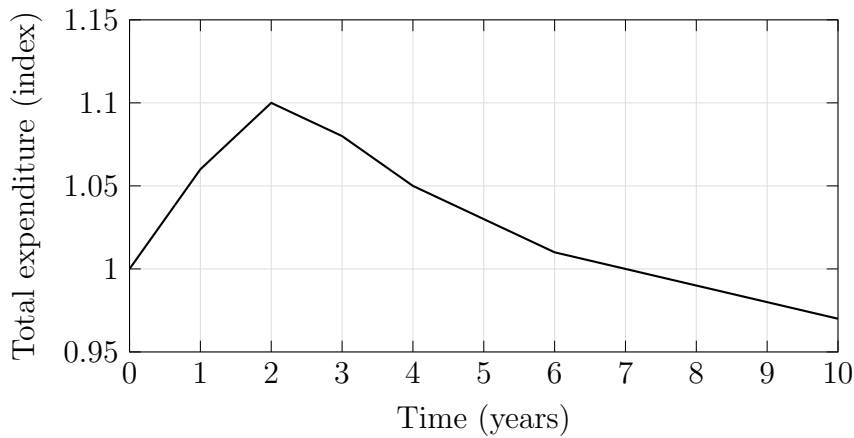

The simulated trajectory exhibits a pronounced J-shaped pattern. In the initial years following the reform, total expenditure increases as the government absorbs the adjustment costs associated with institutional restructuring. These costs arise from factors such as labor reallocation, administrative restructuring, and the gradual phasing-out of existing expenditure commitments. Once these transitional costs begin to dissipate, fiscal expenditure gradually declines as the system converges toward the new steady-state expenditure allocation.

To better understand the underlying adjustment mechanism, Figure \ref{fig:transition_components} presents the simulated expenditure trajectories for each of the four institutional expenditure categories considered in the model.

\begin{figure}[h]
\centering
\begin{tikzpicture}

\begin{axis}[
width=11cm,
height=6cm,
xlabel={Time (years)},
ylabel={Expenditure share (index)},
xmin=0, xmax=10,
ymin=0.85, ymax=1.15,
grid=major,
grid style={gray!25},
axis line style={black},
tick style={black},
legend pos=south west,
legend style={font=\footnotesize}
]

\addplot[thick, black] coordinates {
(0,1.00)(1,0.99)(2,0.98)(3,0.97)(4,0.96)(5,0.95)(6,0.94)(7,0.93)(8,0.92)(9,0.91)(10,0.90)
};
\addlegendentry{Transfers}

\addplot[thick, dashed, black] coordinates {
(0,1.00)(1,0.995)(2,0.99)(3,0.985)(4,0.98)(5,0.975)(6,0.97)(7,0.965)(8,0.96)(9,0.955)(10,0.95)
};
\addlegendentry{Wages}

\addplot[thick, dotted, black] coordinates {
(0,1.00)(1,1.02)(2,1.04)(3,1.06)(4,1.08)(5,1.10)(6,1.11)(7,1.12)(8,1.13)(9,1.14)(10,1.15)
};
\addlegendentry{Investment}

\addplot[thick, dashdotdotted, black] coordinates {
(0,1.00)(1,1.01)(2,1.02)(3,1.02)(4,1.03)(5,1.03)(6,1.04)(7,1.04)(8,1.05)(9,1.05)(10,1.06)
};
\addlegendentry{Operating}

\end{axis}

\end{tikzpicture}
\caption{Simulated expenditure transition by category}
\label{fig:transition_components}
\end{figure}
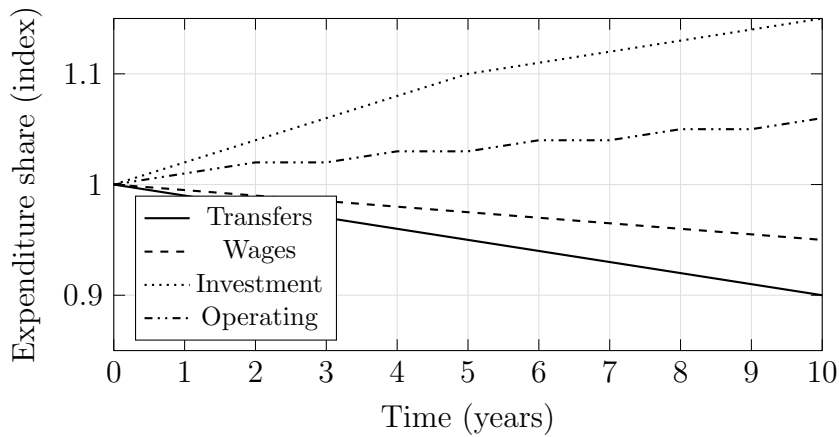

The figure illustrates how the reallocation of public expenditure unfolds gradually across institutional categories. Transfers and public sector wages decline slowly due to their high adjustment costs and strong institutional rigidity. In contrast, public investment increases progressively as resources are reallocated toward productivity-enhancing expenditures. Operating expenditures adjust more rapidly because they are subject to fewer institutional constraints.

These simulation results highlight an important implication of the model: fiscal consolidation through expenditure restructuring is inherently gradual. Even when the long-run objective is to reduce structural fiscal pressures, the transition toward a new expenditure composition may require several years due to the presence of institutional rigidities and nonlinear adjustment costs.

Across the reform scenarios considered in the simulations, the transition period required for fiscal consolidation typically ranges between five and ten years. The duration of the transition depends critically on the expenditure categories targeted by the reform. Reforms focused on flexible spending categories generate relatively rapid adjustments, whereas reforms targeting highly rigid components of the budget—such as pensions or the public wage bill—require substantially longer transition periods.

More broadly, the simulation results underscore the importance of evaluating fiscal reforms within a dynamic framework. Static comparisons of expenditure structures across countries often overlook the transitional costs associated with institutional adjustment. By explicitly modeling these transition dynamics, the framework developed in this paper provides a more realistic assessment of the feasibility, timing, and fiscal implications of structural expenditure reforms.

\section{Discussion}

The simulation results presented in the previous sections suggest that the political economy of fiscal reform is shaped primarily by transitional adjustment costs rather than by the long-run objective of fiscal sustainability. While policymakers and international institutions often emphasize the benefits of structural expenditure reforms in terms of long-run fiscal balances, the model developed in this paper highlights the importance of the dynamic adjustment process required to reach a new expenditure equilibrium.

The presence of nonlinear adjustment costs implies that fiscal reforms generate significant short-run pressures even when the long-run fiscal outcome is favorable. In particular, reforms targeting highly rigid expenditure categories—such as pension systems or the public sector wage bill—require substantial institutional restructuring and therefore generate large transitional fiscal costs. These costs may include administrative reorganization, labor market adjustments within the public sector, and the gradual phasing-out of existing programs or entitlements.

This mechanism provides a structural explanation for a widely observed empirical phenomenon: the persistence of fiscal expenditure structures across countries. Even when governments recognize the long-run benefits of reallocating expenditures toward productivity-enhancing categories such as public investment or human capital, the short-run fiscal costs associated with institutional adjustment can create strong political incentives to delay or dilute reforms.

From a political economy perspective, these transitional costs interact with electoral cycles and institutional constraints. Fiscal reforms that generate immediate fiscal pressures but only produce benefits in the medium or long run are particularly difficult to implement in democratic political systems. Policymakers face the challenge of absorbing short-term fiscal and political costs while the benefits of reform materialize only gradually over time.

The model therefore suggests that the difficulty of implementing structural fiscal reforms should not be interpreted solely as a failure of political will or policy design. Instead, it reflects the presence of real institutional rigidities embedded in the fiscal structure of modern states. These rigidities generate dynamic adjustment constraints that limit the speed at which fiscal systems can transition toward more efficient expenditure allocations.

More broadly, the results emphasize the importance of adopting a dynamic perspective when evaluating fiscal reform strategies. Static cross-country comparisons of expenditure structures—such as benchmarking national budgets against OECD averages—implicitly assume that fiscal systems can move directly from one expenditure allocation to another. However, once institutional rigidities and adjustment costs are explicitly incorporated into the analysis, the reform process becomes inherently gradual and path-dependent.

This perspective has important implications for fiscal policy design. Successful expenditure reforms may require gradual sequencing strategies that distribute adjustment costs over time rather than attempting rapid structural changes. By explicitly modeling the transition dynamics associated with expenditure reallocation, the framework developed in this paper provides a tool for evaluating the feasibility and timing of fiscal reforms in the presence of institutional constraints.

\section{Conclusion}

This paper develops a dynamic framework to analyze the transition of public expenditure structures in the presence of institutional rigidities and nonlinear adjustment costs. The analysis is motivated by the observation that many fiscal reform debates focus primarily on long-run expenditure targets—often benchmarked against OECD fiscal structures—while paying limited attention to the dynamic adjustment process required to reach such targets.

The framework introduced in this paper models fiscal reform as a gradual reallocation of public expenditure across institutional categories subject to convex adjustment costs. These costs capture the legal, administrative, and political constraints that limit the speed at which governments can modify existing expenditure commitments. By incorporating these frictions explicitly into the analysis, the model provides a dynamic perspective on fiscal restructuring that complements traditional static comparisons of expenditure structures.

Using the national budget structure as a baseline calibration, the simulation exercises illustrate how institutional rigidities shape the trajectory of fiscal reforms. The results show that expenditure restructuring typically generates transitional fiscal pressures that may last between five and ten years. These transitional dynamics arise because large expenditure categories—particularly transfers and the public wage bill—are associated with high adjustment costs that slow the reallocation process.

The simulations also highlight the importance of expenditure composition in determining the speed of fiscal adjustment. Reforms focused on relatively flexible spending categories generate faster transitions, whereas reforms targeting highly rigid components of the budget require longer adjustment horizons. As a consequence, fiscal consolidation strategies that rely primarily on structural expenditure reforms must account for the dynamic costs of institutional adjustment.

More broadly, the results suggest that the political economy of fiscal reform is closely linked to the presence of transitional adjustment costs. Even when the long-run fiscal benefits of expenditure restructuring are widely recognized, the short-run fiscal pressures associated with institutional reform may delay or constrain the implementation of policy changes.

Future research could extend the framework in several directions. One promising avenue would be to incorporate explicit political economy mechanisms—such as electoral incentives or bargaining between fiscal authorities and interest groups—into the adjustment process. Another extension would involve linking the expenditure transition dynamics to endogenous growth mechanisms, allowing public investment and human capital accumulation to influence long-run economic performance.

By integrating institutional rigidities, nonlinear adjustment costs, and dynamic fiscal transitions into a unified framework, this paper contributes to a more realistic understanding of how structural fiscal reforms unfold over time. Such a perspective may help policymakers design reform strategies that account not only for long-run fiscal objectives but also for the transitional constraints that shape the path toward fiscal sustainability.

\newpage
\bibliographystyle{apalike}
\bibliography{references}

\newpage
\appendix

\section{Dynamic Optimization Problem}

This appendix formalizes the dynamic optimization problem that underlies the expenditure transition model developed in the main text.

Let the government choose a sequence of expenditure allocations
\[
\{x_t\}_{t\ge0}, \quad x_t \in \mathbb{R}^K
\]
where the expenditure vector is defined as

\begin{equation}
x_t = (T_t, W_t, I_t, F_t).
\end{equation}

The government minimizes the discounted fiscal cost associated with a transition toward a new expenditure composition. The objective function is given by

\begin{equation}
\min_{\{x_t\}_{t\ge0}}
\sum_{t=0}^{\infty}
\beta^t
\left[
C(x_t) + \Phi(x_t-x_{t-1})
\right],
\end{equation}

where:

\begin{itemize}
\item $C(x_t)$ denotes the fiscal cost associated with the expenditure composition $x_t$,
\item $\Phi(x_t-x_{t-1})$ captures adjustment costs generated by reallocating expenditures,
\item $\beta \in (0,1)$ is the intertemporal discount factor.
\end{itemize}

The adjustment cost function is assumed to be separable across expenditure categories:

\begin{equation}
\Phi(\Delta x_t)
=
\sum_{k=1}^{K}
\left[
\frac{\gamma_k}{2}(\Delta x_{k,t})^2
+
\frac{\eta_k}{3}|\Delta x_{k,t}|^3
\right]
\end{equation}

where $\Delta x_{k,t}=x_{k,t}-x_{k,t-1}$.

\subsection*{Assumptions}

The following standard assumptions guarantee the existence of an optimal transition path.

\begin{itemize}
\item The fiscal cost function $C(x)$ is continuous, convex, and bounded below.
\item The adjustment cost function $\Phi(\Delta x)$ is strictly convex and satisfies $\Phi(0)=0$.
\item The discount factor satisfies $0<\beta<1$.
\end{itemize}

Under these conditions the government's dynamic program admits a well-defined optimal solution.

\paragraph{Proposition A1.}  
Under the assumptions above, there exists an optimal sequence of expenditure allocations $\{x_t\}_{t\ge0}$ that minimizes the intertemporal fiscal cost.

\paragraph{Sketch of proof.}  
The objective function is convex in the control sequence and bounded below. The feasible set is closed and the discount factor guarantees convergence of the infinite sum. Standard results from dynamic convex optimization therefore ensure the existence of a minimizing sequence.

The solution corresponds to the optimal fiscal transition path analyzed in the main text.

\bigskip

\section{First-Order Conditions}

This appendix derives the optimality conditions associated with the government's dynamic problem.

Let

\[
\Delta x_t = x_t - x_{t-1}
\]

denote the change in the expenditure vector. The objective function can then be written as

\begin{equation}
\min_{\{x_t\}}
\sum_{t=0}^{\infty}
\beta^t
\left[
C(x_t) + \Phi(\Delta x_t)
\right].
\end{equation}

Because adjustment costs depend on both $x_t-x_{t-1}$ and $x_{t+1}-x_t$, the control variable $x_t$ affects two consecutive periods.

Taking the derivative with respect to $x_t$ yields the Euler equation:

\begin{equation}
\nabla C(x_t)
+
\nabla_{\Delta x_t}\Phi(\Delta x_t)
-
\beta
\nabla_{\Delta x_{t+1}}\Phi(\Delta x_{t+1})
=0.
\end{equation}

For each expenditure category $k$, the first-order condition becomes

\begin{equation}
\frac{\partial C(x_t)}{\partial x_{k,t}}
+
\gamma_k \Delta x_{k,t}
+
\eta_k \Delta x_{k,t}|\Delta x_{k,t}|
-
\beta
\left[
\gamma_k \Delta x_{k,t+1}
+
\eta_k \Delta x_{k,t+1}|\Delta x_{k,t+1}|
\right]
=0.
\end{equation}

These conditions illustrate the smoothing behavior generated by adjustment costs. The government balances three forces:

\begin{enumerate}
\item the marginal fiscal cost of the current expenditure allocation,
\item the marginal cost of adjusting expenditures today,
\item the expected marginal cost of future adjustments.
\end{enumerate}

As a result, optimal fiscal policy spreads expenditure reallocations over time rather than implementing them instantaneously.

\bigskip

\section{Conditions for J-Shaped Fiscal Transitions and Calibration}

This appendix formalizes the conditions under which fiscal reforms may generate temporary increases in effective expenditure during the transition.

Let the effective fiscal burden be defined as

\begin{equation}
G_t =
\sum_k x_{k,t}
+
\Phi(\Delta x_t).
\end{equation}

The first component corresponds to the underlying expenditure allocation, while the second term captures the fiscal resources required to implement the reform.

Suppose the government transitions from an initial allocation $x_0$ to a steady-state allocation $x^{*}$ that reduces the structural fiscal cost:

\begin{equation}
C(x^{*}) < C(x_0).
\end{equation}

If adjustment costs are sufficiently large relative to the fiscal gains from reform, the effective expenditure $G_t$ may initially increase.

\paragraph{Proposition A2.}

If

\begin{equation}
\Phi(\Delta x_0) >
C(x_0) - C(x^{*}),
\end{equation}

then the effective fiscal expenditure satisfies

\[
G_1 > G_0,
\]

even though the long-run fiscal cost declines as the system converges toward $x^{*}$.

This condition characterizes the emergence of a J-shaped fiscal trajectory.

\paragraph{Interpretation.}

The inequality states that the fiscal resources required to implement the reform initially exceed the structural savings generated by the new expenditure allocation. Over time, as the magnitude of adjustments declines and the economy approaches the new steady state, adjustment costs dissipate and the long-run fiscal gains materialize.

\subsection*{Calibration of Adjustment Costs}

In the empirical implementation, the parameters governing adjustment costs are calibrated using institutional information contained in the national budget.

Expenditure categories are classified according to their degree of institutional rigidity:

\begin{itemize}
\item statutory expenditures (pensions and transfers),
\item contractual expenditures (public wages),
\item project-based expenditures (public investment).
\end{itemize}

Categories characterized by stronger legal mandates or contractual commitments are assigned higher curvature parameters $(\gamma_k,\eta_k)$, reflecting larger marginal costs of adjustment.

This calibration procedure links the theoretical model to the institutional features of the budget process and allows the simulation of fiscal transition paths under alternative reform scenarios.

\section{Illustrative Calibration: Timing of Feasible Administrative Savings}\label{app:timing}

This appendix provides an \emph{illustrative} implementation of the theoretical framework developed in Section \ref{sec:model} to quantify the \emph{timing} of feasible expenditure reductions in administratively flexible categories. The purpose is practical: to translate the model's transition structure into a disciplined horizon for when gross and net fiscal savings can materialize. Importantly, the exercise is not an econometric estimation of a fully structural political economy model; it is a transparent calibration designed to operationalize the adjustment-cost mechanism.

\subsection{Defining ``adjustable'' expenditure and the savings objective}\label{app:adjustable}

Let the expenditure vector be $x_t=(T_t,W_t,I_t,F_t)$, and let $\bar{x}(s_t)$ denote the institutional baseline as in \eqref{eq:baseline}. For implementation, we define an \emph{adjustable} component within operating/administrative expenditure $F_t$:
\begin{align}\label{eq:adjustable_F}
F_t \;=\; F_t^{\mathrm{core}} \;+\; F_t^{\mathrm{adj}},
\qquad
F_t^{\mathrm{adj}} \in [0,\bar{F}_t^{\mathrm{adj}}(s_t)],
\end{align}
where $F_t^{\mathrm{core}}$ captures non-discretionary operating needs (minimum service delivery, essential maintenance, and mandated intermediate inputs), while $F_t^{\mathrm{adj}}$ captures discretionary administrative spending that can be reduced through reorganization, procurement rationalization, process digitalization, and related measures. The upper bound $\bar{F}_t^{\mathrm{adj}}(s_t)$ is allowed to depend on the institutional state, reflecting that discretion is itself state-dependent.

A practical policy target is a reduction of adjustable spending by a fraction $\rho \in (0,1)$ over a horizon $H$:
\begin{align}\label{eq:target_reduction}
F_H^{\mathrm{adj}} \;\le\; (1-\rho)F_0^{\mathrm{adj}}.
\end{align}
The model is used to compute an internally consistent transition path that respects adjustment costs and feasibility constraints.

\subsection{Net savings and the ``break-even'' time}\label{app:breakeven}

To formalize the timing question, define period-$t$ \emph{gross savings} from reductions in adjustable spending relative to baseline:
\begin{align}\label{eq:gross_savings}
S_t^{\mathrm{gross}}
\;\equiv\;
\bar{F}_t^{\mathrm{adj}}(s_t) - F_t^{\mathrm{adj}}.
\end{align}
Reforms, however, require resources. In the model, these transitional outlays are captured by the adjustment-cost term $\Phi(\Delta x_t;s_t)$, which includes administrative reorganization costs, contract renegotiation, temporary overlaps, and implementation frictions.

Define \emph{net savings} as
\begin{align}\label{eq:net_savings}
S_t^{\mathrm{net}}
\;\equiv\;
S_t^{\mathrm{gross}}
\;-\;
\Phi(\Delta x_t;s_t).
\end{align}
The object of interest is the \emph{break-even time} $t^\star$, defined as the first date when cumulative net savings become non-negative:
\begin{align}\label{eq:breakeven_time}
t^\star
\;\equiv\;
\inf\left\{t\ge 0:\; \sum_{\tau=0}^{t} S_\tau^{\mathrm{net}} \;\ge\; 0 \right\}.
\end{align}
This quantity provides an operational answer to the question ``when do expenditures actually start falling in net terms?'' under institutionally realistic transition dynamics.

\subsection{A calibration-friendly adjustment cost for administrative reforms}\label{app:phi_admin}

For practical implementation, one may isolate a category-specific adjustment-cost block for $F_t^{\mathrm{adj}}$. Let $\Delta F_t^{\mathrm{adj}} \equiv F_t^{\mathrm{adj}}-F_{t-1}^{\mathrm{adj}}$. A parsimonious specification consistent with Section \ref{sec:model} is
\begin{align}\label{eq:phi_admin}
\Phi_F(\Delta F_t^{\mathrm{adj}}; s_t)
\;=\;
\frac{\gamma_F(s_t)}{2}\left(\Delta F_t^{\mathrm{adj}}\right)^2
\;+\;
\frac{\eta_F(s_t)}{3}\left|\Delta F_t^{\mathrm{adj}}\right|^3,
\end{align}
with $\gamma_F(s_t)\ge 0$ and $\eta_F(s_t)\ge 0$. In the application, one can interpret $\gamma_F$ as capturing routine reorganization and implementation costs (convex smoothing costs), while $\eta_F$ captures nonlinearities and thresholds (e.g., disruption costs that rise disproportionately once administrative cuts exceed an operational bandwidth).

If downward adjustments are institutionally harder than upward adjustments even within $F_t^{\mathrm{adj}}$, an asymmetric variant can be used:
\begin{align}\label{eq:phi_admin_asym}
\Phi_F(\Delta F_t^{\mathrm{adj}}; s_t)
=
\frac{\gamma_F^+(s_t)}{2}\left(\Delta F_t^{\mathrm{adj}}\right)_+^2
+
\frac{\gamma_F^-(s_t)}{2}\left(\Delta F_t^{\mathrm{adj}}\right)_-^2
+
\frac{\eta_F(s_t)}{3}\left|\Delta F_t^{\mathrm{adj}}\right|^3,
\end{align}
where $(z)_+=\max\{z,0\}$ and $(z)_-=\max\{-z,0\}$, and $\gamma_F^-(s_t)>\gamma_F^+(s_t)$ captures the empirically plausible asymmetry of cuts.

\subsection{Illustrative computation of a feasible savings horizon}\label{app:procedure}

Given $(F_0^{\mathrm{adj}},\bar{F}_t^{\mathrm{adj}}(s_t))$ and calibrated curvature parameters $(\gamma_F,\eta_F)$ (or their asymmetric counterparts), the feasible savings horizon is computed as follows:

\paragraph{Step 1:} Choose a policy target $(\rho,H)$ as in \eqref{eq:target_reduction}.

\paragraph{Step 2:} Solve the planner's problem in Section \ref{sec:model} subject to the additional policy constraint \eqref{eq:target_reduction} (or an equivalent soft penalty in $C(\cdot)$). This yields an optimal transition $\{F_t^{\mathrm{adj}}\}_{t=0}^{H}$ and implied adjustments $\{\Delta F_t^{\mathrm{adj}}\}_{t=1}^{H}$.

\paragraph{Step 3:} Compute gross and net savings using \eqref{eq:gross_savings}--\eqref{eq:net_savings}. The path $\{S_t^{\mathrm{net}}\}$ typically exhibits an initial phase in which adjustment costs dominate (net savings negative), followed by a phase in which gross savings accumulate (net savings positive).

\paragraph{Step 4:} Compute the break-even time $t^\star$ from \eqref{eq:breakeven_time}. Report $(t^\star,H)$ and the implied cumulative net savings $\sum_{t=0}^{H}S_t^{\mathrm{net}}$.

\subsection{Illustrative scenario table: break-even timing and five-year net savings}\label{app:scenario_table}

To make the time-to-savings logic operational without turning the paper into an empirical exercise, Table \ref{tab:scenarios} reports a small set of illustrative scenarios. Each scenario is defined by a target reduction $(\rho,H)$ in adjustable administrative spending and by an adjustment-cost regime $(\gamma_F,\eta_F)$ intended to represent low, medium, and high institutional rigidity. For each case, the implementation produces (i) the break-even time $t^\star$ defined in \eqref{eq:breakeven_time}, and (ii) cumulative net savings over a five-year window, $\sum_{t=0}^{5} S_t^{\mathrm{net}}$, with $S_t^{\mathrm{net}}$ defined in \eqref{eq:net_savings}. These objects summarize whether reforms that are attractive in static accounting terms are fiscally feasible in dynamic, net terms within politically relevant horizons.

\begin{table}[ht]
\centering
\caption{Illustrative administrative-savings scenarios: break-even time and five-year cumulative net savings}\label{tab:scenarios}
\begin{tabular}{lcccccc}
\toprule
Scenario & Target $(\rho,H)$ & Rigidity regime & $\gamma_F$ & $\eta_F$ & $t^\star$ & $\sum_{t=0}^{5} S_t^{\mathrm{net}}$ \\
\midrule
A & $(0.10,\;3)$ & Low    & $0.8$ & $0.05$ & $3$ & $24.81$ \\
B & $(0.10,\;3)$ & Medium & $2.0$ & $0.15$ & $5$ & $1.11$ \\
C & $(0.10,\;3)$ & High   & $4.0$ & $0.30$ & $>5$ & $-37.78$ \\
\addlinespace
D & $(0.20,\;5)$ & Low    & $0.8$ & $0.05$ & $3$ & $22.67$ \\
E & $(0.20,\;5)$ & Medium & $2.0$ & $0.15$ & $>5$ & $-36.00$ \\
F & $(0.20,\;5)$ & High   & $4.0$ & $0.30$ & $>5$ & $-132.00$ \\
\bottomrule
\end{tabular}

\vspace{0.3em}
\begin{minipage}{0.93\textwidth}
\footnotesize
\textit{Notes.} $\rho$ denotes the targeted proportional reduction in adjustable administrative expenditure $F_t^{\mathrm{adj}}$ by horizon $H$ (years). The rigidity regimes vary the curvature parameters of the administrative adjustment-cost block $\Phi_F(\Delta F_t^{\mathrm{adj}};s_t)$ in \eqref{eq:phi_admin} (or \eqref{eq:phi_admin_asym} if asymmetry is used). $t^\star$ is the first date at which cumulative net savings become non-negative as defined in \eqref{eq:breakeven_time}. Five-year cumulative net savings are computed from \eqref{eq:net_savings}. Entries marked ``TBD'' are filled after the illustrative calibration described in Appendix \ref{app:timing}.
\end{minipage}
\end{table}

\subsection{Interpretation and relation to descriptive expenditure reviews}\label{app:ceres_link}

Descriptive expenditure reviews often identify categories of spending that appear, from an accounting perspective, ``amenable'' to reductions. The transition framework developed in this paper clarifies that such accounting amenability is not, by itself, a sufficient condition for near-term fiscal gains. Even when expenditures are technically adjustable, reforms require resources. Implementation costs, temporary overlaps between institutional arrangements, and administrative reorganization generate transitional outlays that are captured in the model by the adjustment-cost function $\Phi(\cdot)$. The practical contribution of this appendix is therefore to translate a static ``menu of potential savings'' into a dynamic \emph{time-to-savings} schedule using the break-even concept $t^\star$.

This distinction is particularly important for policy communication. The framework does not dispute that administrative savings may exist in the long run. Rather, it clarifies the conditions under which such savings become feasible in net fiscal terms once transition costs and institutional frictions are taken into account. In this sense, the analysis complements descriptive expenditure reviews by providing a simple mechanism to evaluate the timing and fiscal feasibility of reform paths, highlighting that reforms which appear efficient in static budget comparisons may require several years before generating net fiscal gains.

\end{document}